\documentclass[letterpaper, 10 pt, conference]{ieeeconf}  

\IEEEoverridecommandlockouts                              

\usepackage{graphics} 
\usepackage{epsfig} 
\usepackage{amsmath} 
\usepackage{amssymb}  
\usepackage{amsfonts}
\usepackage{amsthm}
\usepackage{algorithm}

\usepackage{algorithmic}
\usepackage{booktabs}
\usepackage[footnotesize,bf]{caption}
\usepackage[hidelinks]{hyperref}
\usepackage{cleveref}
\usepackage{mathtools}
\usepackage{hyperref}
\hypersetup{colorlinks,
            linkcolor=blue,
            citecolor=blue,
            urlcolor=magenta,
            linktocpage,
            plainpages=false}

\DeclareMathOperator*{\argmax}{arg\,max}

\newtheorem{definition}{Definition}
\newtheorem{assumption}{Assumption}
\newtheorem{lemma}{Lemma}

\DeclareMathOperator{\diag}{diag}
\DeclareMathOperator{\rank}{rank}

\renewcommand{\diag}{\mathop{\mathrm{diag}}}

\newcommand{\arxiv}{}

\newcommand{\colored}[1]{\ifdefined\arxiv #1\else {\color{blue}#1}\fi}

\title{\LARGE \bf
Tube-Based Zonotopic Data-Driven Predictive Control 

}

\crefformat{section}{\S#2#1#3}
\crefformat{subsection}{\S#2#1#3}
\crefformat{subsubsection}{\S#2#1#3}
\crefrangeformat{section}{\S\S#3#1#4 to~#5#2#6}
\crefmultiformat{section}{\S\S#2#1#3}{ and~#2#1#3}{, #2#1#3}{ and~#2#1#3}
\author{Alessio Russo$^{\star,1}$ and Alexandre Proutiere$^{1}$
\thanks{$^{\star}$ Corresponding author.}
\thanks{$^{1}$Alessio Russo and Alexandre Proutiere are in the Division of Decision and Control Systems of the EECS School at KTH Royal Institute of Technology, Stockholm, Sweden.
        {\tt\small \{alessior,alepro\}@kth.se}}
}

\newcommand{\zonotope}[1]{\mathcal{Z}_{#1}}
\newcommand{\matrixzonotope}[1]{\mathcal{M}_{#1}}

\newcommand{\constrainedmatrixzonotope}[1]{\mathcal{N}_{#1}}
\newcommand{\data}{\mathcal{D}}
\begin{document}

\maketitle
\thispagestyle{empty}
\pagestyle{empty}

\begin{abstract}
We present a novel tube-based data-driven predictive control method for linear systems affected by a bounded addictive disturbance. Our method leverages recent results in the reachability analysis of unknown linear systems to formulate and solve a robust tube-based predictive control problem. More precisely, our approach consists in deriving, from the collected data, a zonotope that includes the true state error set. \colored{We show how to guarantee the stability of the resulting error zonotope, which can be exploited to increase the computational efficiency of existing zonotopic data-driven MPC formulations}. Results on a double-integrator affected by strong adversarial noise demonstrate the effectiveness of the proposed control approach.
\end{abstract}

\section{Introduction}
A recent trend in the control community is that of controlling unknown dynamical systems solely from their input-output data. This trend, which has sparked from a growing interest in machine learning and reinforcement learning methods,  is more commonly known as direct control, or model-free control, and has already been investigated in different ways, see for example \ifdefined \arxiv Direct Adaptive Control \cite{ioannou1986robust},  data-driven control methods such as  VRFT \cite{campi2002virtual}, and other references in \cite{markovsky2008data}. \else \cite{ioannou1986robust,campi2002virtual,markovsky2008data}.\fi

In addition to classical direct control methods, recent research has made use of some results in \cite{willems2005note}, that allow to characterize an unknown linear system through a finite collection of its input/output data (which is also known as Willem's et al. lemma, or \textit{fundamental} lemma, \cite{willems2005note,markovsky2008data}).
Thanks to this representation, it is possible to directly analyze the system, or formulate a control law, using only the collected input-output data, without the need of identifying the underlying unknown system. For example, one may formulate data-driven versions of the linear quadratic regulator \cite{de2019formulas}, or of model-predictive control (MPC) \cite{coulson2019data}.

Data-driven MPC formulations based on the fundamental lemma have several advantages compared to traditional MPC: they may require less data compared to classical learning-based approaches, and there is no need to identify the underlying system, which could be a costly process for complex systems. However, data-driven MPC formulations based on the fundamental lemma tend to be brittle, and several studies in the literature  thoroughly investigate the robustness  of data-driven MPC procedures by analysing the resulting multiplicative model uncertainty in the Hankel matrices of the system. \cite{coulson2019regularized,huang2021robust,breschi2022role,berberich2020data}. 

In contrast, in this work, we address the problem of robustness for data-driven predictive control by exploiting the data-driven zonotopic reachability analysis proposed in \cite{alanwar2021data}. Reachability analysis computes the set of trajectories that a system can reach in a finite amount of time, and it is used for formal verification and set-based estimation due to its robust control guarantees \cite{althoff2010reachability,althoff2021set}. Following the analysis in \cite{alanwar2021data}, our approach is based on providing robustness with respect to all possible system matrices that are consistent with the data collected by the user. Our method is inspired by \cite{alanwar2022robust}. There, the authors provide a robust data-driven predictive control procedure based on reachability analysis. This procedure however tends to be computationally intensive as well as sensitive to potential instabilities.

We propose a tube-based approach to robust data-driven predictive control. We first discuss how min-max robust control \cite{alanwar2022robust} tends to be computationally infeasible for uncertainties represented by zonotopes, and then propose our method. The underlying idea is to treat the unknown quantities as external disturbances of the system, while guaranteeing the stability of the reachable error trajectories. We show that for a stabilizing gain $K$ it is possible to bound the resulting error zonotope of the system. Identifying a stabilizing gain $K$ is NP-hard. We present simple methods, with probabilistic guarantees, to identify $K$ and to verify that it is stabilizing. We conclude by showing our method on a double-integrator affected by strong-noise, demonstrating how our procedure can guarantee robustness and constraint satisfaction.
\section{Related work}
The problem of robustness for MPC has been extensively studied in the literature, see also the following standard references \cite{lofberg2003minimax,rawlings2017model}. Standard min-max formulations usually assume uncertainty in the process noise, and are usually solved by means of semi-definite relaxations \cite{lofberg2003minimax}. 

As an alternative, tube-based approaches \cite{langson2004robust,mayne2005robust} tend to be more computationally modest, while being robust to all disturbance sequences. The goal of tube-based MPC is to ensure that the nominal trajectory of the system satisfies tightened constraints, so that all possible trajectories of the true system satisfy the original constraint. 

\colored{In this work we do not make any statistical assumption on the noise, and assume that the system matrices $(A,B)$ are uncertain}. A similar problem has been studied in \colored{\cite{de2004robust,calafiore2012robust,lu2019robust,bujarbaruah2021simple}}. In \cite{de2004robust} the authors consider the classical quadratic criterion on the state-action pair, and assume a \colored{noise-less} linear system  with polytopic uncertainty in the system matrices $(A,B)$. Similarly, \cite{calafiore2012robust} considers Scenario Optimization \cite{campi2008exact} to handle generic bounded uncertainty sets. In \cite{lu2019robust} \colored{the authors consider a similar problem, with additive disturbance and parametric uncertainty. They assume  the system matrices to be affine functions, \emph{i.e.}, $(A(\theta),B(\theta))=(A_0,B_0)+\sum_{i=1}^p (A_i,B_i)\theta_i$, for known matrices $\{(A_i,B_i)\}_{i=0}^p$ and some parameter $\theta\in\Theta_0\subset \mathbb{R}^p$, with $\Theta_0$ being a known bounded convex polytope. In contrast, we consider a data-driven approach to estimate the set of system matrices compatible with the data.  Finally, the authors in \cite{bujarbaruah2021simple} consider a setting similar to that of \cite{lu2019robust}. The authors assume the system matrices to belong to convex compact sets with \textit{known} vertices, whereas we work with sets defined using matrix zonotopes.}

Recently, the authors of \cite{alanwar2022robust} proposed \textsc{ZPC}, a robust data-driven predictive control approach based on data-driven reachability analysis \cite{alanwar2021data} to control an uncertain linear system affected by bounded noise. The input-output data of the system is used to construct a matrix zonotope that contains all the possible matrices $(A,B)$ that are consistent with the data, which is then used to formulate a robust MPC approach based on the reachable states. Their formulation, however, is not a tube-based approach, and does not consider the problem of instability in the set of matrices that are consistent with the data. Other  robust formulations of data-driven predictive control approaches analyze the robustness of data-driven MPC from a different perspective, mainly by analyzing the resulting multiplicative model uncertainty in the Hankel matrices of the system \cite{huang2021robust,berberich2020data}.
\section{Problem Statement and Preliminaries}
We first provide some preliminary concepts on set representation,  and then state our control problem.
\subsection{Set Representations}

\begin{definition}[Zonotope \cite{kuhn1998rigorously}]
A Zonotope $\mathcal{Z}$ of dimension $n$, with $\gamma$ generators, is a set defined as
\begin{equation}
    \mathcal{Z} = \left\{x\in \mathbb{R}^n : x=c_\mathcal{Z}+G_\mathcal{Z} \beta, \|\beta\|_\infty \leq 1, \beta\in \mathbb{R}^\gamma \right\},
\end{equation}
where  $c_\mathcal{Z} \in \mathbb{R}^n$ is the center, and  $G_\mathcal{Z} =\begin{bmatrix}g_\mathcal{Z} ^{(1)},\dots, g_\mathcal{Z} ^{(\gamma)} \end{bmatrix}\in \mathbb{R}^{n\times \gamma}$, is the generator matrix. Furthermore, we define the shorthand $\mathcal{Z}=\langle c_\mathcal{Z},G_\mathcal{Z}\rangle$.
\end{definition}
Zonotopes are special polytopes, and are widely used in reachability analysis \cite{althoff2010reachability} due to their compact representation. Their images through linear mappings and their Minkowski sums can be efficiently computed. A linear mapping is defined as $T\zonotope{}=\{Tz: z\in \zonotope{}\}$ and the Minkowski sum between two zonotopes $\zonotope{1}, \zonotope{2}$ is computed as $\zonotope{1}\oplus \zonotope{2} = \langle c_{\zonotope{1}} + c_{\zonotope{2}}, \begin{bmatrix} G_{\zonotope{1}}, G_{\zonotope{2}}\end{bmatrix} \rangle$. For simplicity, we denote the sum using the $+$ sign instead of $\oplus$. \colored{Similarly, we use $\zonotope{1}-\zonotope{2}$ to denote $\zonotope{1}+(-1\zonotope{2})$.} We can also define the concept of matrix zonotope, which is a set of matrices.
\begin{definition}[Matrix Zonotope \cite{althoff2010reachability}]
A Matrix Zonotope $\mathcal{M}$ of dimension $(n,p)$, with $\gamma$ generators, is a set defined as
\begin{equation}
    \matrixzonotope{} = \left\{X\in \mathbb{R}^{n\times p} : X=C_\mathcal{M}+\sum_{i=1}^\gamma G_\mathcal{M}^{(i)} \beta_i, \|\beta\|_\infty \leq 1 \right\},
\end{equation}
where  $C_\mathcal{M} \in \mathbb{R}^{n\times p}$ is the center, and  $G_\mathcal{M} =\begin{bmatrix}G_\mathcal{M} ^{(1)},\dots, G_\mathcal{M} ^{(\gamma)} \end{bmatrix}\in \mathbb{R}^{n\times p\gamma}$ is the generator matrix and $\beta\in \mathbb{R}^{\gamma}$ is the generator factor. \colored{We define the following shorthand for matrix zonotopes $\mathcal{M}=\langle C_\mathcal{Z},G_\mathcal{M}\rangle$.}
\end{definition}
\colored{A linear mapping is defined as $T\mathcal{M}=\{TX: X\in \mathcal{M}\}$ (similarly, $T\mathcal{M}$).} We define the concatenation of two zonotopes, which is the horizontal stacking of two matrix zonotopes $
    \matrixzonotope{AB} = \left\{\begin{bmatrix}X_A & X_B\end{bmatrix}: X_A \in \matrixzonotope{A}, X_B\in \matrixzonotope{B}\right\}
$. From this definition, we let $\matrixzonotope{A^T}$ be the concatenation of a matrix zonotope $\matrixzonotope{A}$ with itself $T$ times, \emph{i.e.}, $\matrixzonotope{A^T} = \left\{\begin{bmatrix}X_A^{(1)} & X_A^{(2)}&\cdots &X_A^{(T)}\end{bmatrix}: X_A^{(i)} \in \matrixzonotope{A},i=1,\dots,T\right\}$.

\subsection{Problem Statement}
\noindent\textbf{Model.} We consider an uncertain discrete-time LTI model affected by process noise:
\begin{equation}\label{eq:system}
x_{t+1}=A_0x_t+B_0u_t+w_t,
\end{equation}
where $t \in \mathbb{Z}$ is the discrete time variable, $x_t\in \mathbb{R}^n$ is the state of the system, $u_t\in \mathbb{R}^m$ is the control signal, $A_0\in\mathbb{R}^{n\times n}, B_0\in \mathbb{R}^{n\times  m}$ are the unknown system matrices, and $w_t\in \mathbb{R}^n$ is the process noise.
We  make the following assumption of boundedness on the process noise, which does not necessarily need to be i.i.d.
\begin{assumption}
The process noise $w(t)$ lies in $\zonotope{w}=\langle c_{\zonotope{w}}, G_{\zonotope{w}}\rangle$, \emph{i.e.}, $w(t) \in \zonotope{w} \subset \mathbb{R}^n$ for every $t$. Furthermore, we assume that $0\in \zonotope{w}$. We denote by $\gamma_w$ the number of generators of $\zonotope{w}$.
\end{assumption}
The objective is to robustly control the uncertain system in \Cref{eq:system} for all possible noise realizations $w_t\in \zonotope{w}$. The pair $(A_0,B_0)$ is unknown, and we use data to develop a control algorithm. For a given trajectory $\{(u_k,x_k)\}_k$ of length $T$, define the following matrices:
\begin{align*}
X_+ &\coloneqq \begin{bmatrix}x_1 & \dots & x_T\end{bmatrix}, \quad X_-\coloneqq \begin{bmatrix}x_0 & \dots & x_{T-1}\end{bmatrix},\\
U_- &\coloneqq \begin{bmatrix}u_0 & \dots & u_{T-1}\end{bmatrix}.
\end{align*}

We make the following assumption, which states that persistent excitation is present in the data \cite{willems2005note}.

\begin{assumption}
The pair $(A_0,B_0)$ is unknown, and  the decision maker  has available one input-state trajectory $\mathcal{D}=(X_-,X_+,U_-)$ such that $\rank \begin{bmatrix} X_- \\ U_-\end{bmatrix}=n+m$.
\end{assumption}

The rank condition can be verified directly from data, and can be guaranteed for noise-free systems by choosing a persistently exciting input signal of order $n+1$ \cite{willems2005note,de2019formulas}. Since the pair $(A_0,B_0)$ is unknown, as well as the actual realization of the noise, there exist multiple pairs $(A,B)$ that are consistent with the data. We denote this set by $\Sigma_\mathcal{D}$:
\[\Sigma_\mathcal{D} \coloneqq \{(A,B): X_+ = AX_- + BU_-+ W_-, W_-\in \zonotope{w}\}.\]
\noindent \textbf{Problem statement.}
Taking inspiration from \cite{lofberg2003minimax,bravo2006robust,calafiore2012robust,alanwar2022robust}, our objective is to robustly control the unknown system in \Cref{eq:system} using a receding horizon approach.

Specifically, the control objective is to minimize a sum of convex loss functions $\{\ell_k(x,u)\}_k$ over an $M$-steps horizon, while constraining the state of the system $x_t$ to  a \colored{bounded zonotope $\zonotope{x}$} at each time step, and the control signal $u_t$ to a \colored{bounded zonotope $\zonotope{u}$}. Furthermore, we assume that the initial condition \colored{$x_0$ belongs to  $\zonotope{x}$}. Finally, we solve the problem by using a receding horizon algorithm that at each iteration computes the optimal control signal over an horizon of $N\leq M$ steps.
Our approach consists of two phases: (1) an offline data-collection phase, to construct a set approximating and containing $\Sigma_\data$, a set of possible models consistent with the data collected from the true system; (2) an online control phase that solves a robust tube-based MPC problem.

\section{Method}
We start by presenting the first offline phase. It consists in collecting data from the true system and in building a matrix zonotope $\matrixzonotope{\data}$ that contains the  set $\Sigma_\data$. We then describe our online robust control problem and present a computationally efficient approach to solve it.

\subsection{Offline Learning Phase}
In the offline learning phase, we gather in ${\cal D}$ a system trajectory of length $T$, and construct the uncertainty set $\Sigma_\data$ using zonotopes.
Let $\matrixzonotope{\zonotope{w}^T}$ be the $T$-concatenation of the noise zonotope $\zonotope{w}$. From this $T$-concatenation, we can build a matrix zonotope ${\cal M}_\data$ containing $\Sigma_\mathcal{D}$. 
\begin{lemma}[Lemma 1 in \cite{alanwar2021data}]\label{lemma:sys_representation}
Given an input-state trajectory $\data$ of the system \cref{eq:system}, with the matrix $\begin{bmatrix}
X_-^\top & U_-^\top\end{bmatrix}^\top$ having full column rank, then $\Sigma_\data \subseteq \matrixzonotope{\data}$, where $\matrixzonotope{\data}$ is a  matrix zonotope defined as follows:
\begin{equation}
    \matrixzonotope{\data}=(X_+-\matrixzonotope{\zonotope{w}^T}) \begin{bmatrix}
    X_-\\ U_-
    \end{bmatrix}^\dagger.
\end{equation}
\end{lemma}
\noindent Note that \cite{alanwar2021data} provides a precise characterization of $\Sigma_\data$, \emph{i.e.}, it is possible to derive a constrained matrix zonotope  $\constrainedmatrixzonotope{\data}$ that is equal to $\Sigma_\data$. However, $\constrainedmatrixzonotope{\data}$ is not easy to use in practice, and an approximate set containing $\Sigma_\data$ is needed, which motivates the use of $\matrixzonotope{\data}$. 
\subsection{Online Control Phase}
\noindent We first state our min-max robust control problem, and quantify its computational complexity. We then present our tube-based solution approach and explain how it addresses the aforementioned complexity issue.   

\subsubsection{Min-max robust control and its complexity}
Consider the following min-max  optimization problem over an horizon of $N$ steps and the uncertain set $\mathcal{F}_\data^N\coloneqq \{\colored{(A,B,w_0,\dots,w_{N-1}): (A,B)\in \matrixzonotope{\data}, w_i \in \zonotope{w}}\}$ :
\begin{equation}\label{problem:p_x}
\begin{aligned}
\bar{\mathcal{P}}_N(x_t): &\min_{u_{\colored{0}|t},\dots,u_{\colored{N-1}|t}} \max_{\colored{(A,B,w_0,\dots)}\in \mathcal{F}_\data^N}\quad  \sum_{k=\colored{0}}^{\colored{N-1}} \ell_k(x_{k|t},u_{k|t})\\
\textrm{s.t.} \quad 
&x_{k+1|t}=Ax_{k|t}+Bu_{k|t}+w_{k}\quad x_{\colored{0}|t}=x_t,\\
& x_{k+1|t} \in \zonotope{x}, u_{k|t} \in \zonotope{u},\quad \colored{k=0,\dots, N-1.}
\end{aligned}
\end{equation}
In general, solving \eqref{problem:p_x} is computationally prohibitive. Indeed,  the computational complexity (in number of floating operations) of the inner maximization problem scales at least as the number of vertices of the zonotope $\mathcal{F}_\data^N$. This number may in the worst case scale as $O((T\gamma_w-1)^{n(n+m)-1} + 2^{N\gamma_w})$ for varying $T$ and $N$. 

\begin{lemma}\label{lemma:complexity_min_max}
The inner maximization in \eqref{problem:p_x} amounts to checking at most $2\left( \sum_{i=0}^{n(n+m)-1}  {T\gamma_w -1 \choose i} +\sum_{i=0}^{nN-1} {N\gamma_w -1 \choose i} \right)$ points in $\mathcal{F}_\data^N$.\ifdefined\arxiv\else\footnote{\colored{ Refer to technical report \href{https://arxiv.org/abs/2209.03500}{https://arxiv.org/abs/2209.03500} for all the proofs and other details}\label{techreport}.}\fi
\end{lemma}
\ifdefined\arxiv
\begin{proof}
Since the uncertainties $(A,B,w_k)$ belong to a polytope, by linearity, for fixed control inputs, these uncertainties generate  set of predictions that is a polytope. Therefore, we only need to check the vertices of this polytope to compute the inner maximization  problem. The matrix zonotope $\matrixzonotope{\data}$ is of dimensionality $n^2+nm$, and consists of $\gamma_w T$ generators: consequently  $\matrixzonotope{\data}$ consist at most of $2 \sum_{i=0}^{n(n+m)-1}  {T\gamma_w -1 \choose i} \leq 2^{T\gamma_w}$ vertices \cite[Thm. 3.1]{ferrez2005solving}. Similarly, the matrix zonotope $\{(w_1,\dots, w_N): w_i \in \zonotope{w}, i=1,\dots, N\}$ has at-most $2\sum_{i=0}^{nN-1} {N\gamma_w -1 \choose i}\leq 2^{N\gamma_w}$ vertices.
\end{proof}
\else
\fi
We remark that it is possible to partly simplify the complexity issue by {\it over-approximating} the matrix zonotope $\matrixzonotope{\data}$ by an hypercube. By doing so, $\matrixzonotope{\data}$ can be approximated by an hypercube with $2^{n(n+m)}$ vertices (if $\gamma_w T \geq n(n+m)$). However, we cannot address the complexity issue arising due to the set $\{(w_1,\dots, w_N): w_i \in \zonotope{w}, i=1,\dots, N\}$ in a similar way. In fact, the resulting zonotope would have a number of vertices that scales exponentially in $N\gamma_w$, which remains computationally hard when $N$ is not small. To address the complexity issue, we advocate that a tube-based approach may achieve comparable level of robustness and performance, while being computationally more efficient. In fact, by using a tube-based approach we are able to remove the dependency on $N$.

\subsubsection{Tube-based robust control}
In tube-based MPC \cite{mayne2005robust} the problem  (\ref{problem:p_x}) is relaxed by not considering the actual worst realization of the noise sequence $w_t$.
 The idea is to control some nominal dynamics $\bar x_t$ of the system, and to make sure that the error $e_t=x_t-\bar x_t$ is bounded. Our approach consists in devising an algorithm that can take advantage of the theory of zonotopes to guarantee robustness, while making sure that the resulting error zonotope of $e_t$ is bounded in time. We begin by considering the nominal dynamics of the system.\vspace{4pt}

\noindent \textbf{Nominal and error dynamics.} Consider some nominal, user-chosen matrices $(\bar A,\bar B) \in \matrixzonotope{\data}$, and define the nominal predictive dynamics $\bar x_t$ and error signal $e_t$ as:
\begin{equation}\label{eq:nominal_dynamics}
    \bar x_{t+1} = \bar A \bar x_t + \bar B \bar u_t,\quad e_t = x_t-\bar x_t,
\end{equation}
where $\bar u_t$ is the nominal control signal, which is computed by the receding horizon algorithm. Note that, as shown later, stability-wise it is important that $(\bar A,\bar B) $ belong to $\matrixzonotope{\data}$.

We write the true matrices $(A_0,B_0)\in \Sigma_{\data}$ as $A_0=\bar A+\Delta A_0$ and $B_0=\bar B+\Delta B_0$ for some $(\Delta A_0,\Delta B_0)$. Then, \eqref{eq:system} is equivalent to: 
\begin{equation}
    x_{t+1} = \bar A x_t + \bar B u_t + w_t + \Delta A_0 x_t + \Delta B_0 u_t.
\end{equation}
We treat $w_t + \Delta A_0 x_t + \Delta B_0 u_t$ as an additive disturbance of the system. We consider a control signal $u_t$  defined as
\begin{equation}
    u_t = Ke_t + \bar u_t,
\end{equation}
where the gain matrix $K\in \mathbb{R}^{m\times n}$ is used to stabilize the error dynamics. Then, we can derive the dynamics of the error $
    e_{t+1} = (\bar A + \bar B K)e_t + \Delta A_0(e_t + \bar x_t) +\Delta B_0 u_t + w_t.
$
From the latter expression, we deduce that $e_t$ belongs to a well-defined zonotope $\zonotope{e,t}$, \emph{i.e.}, $e_t\in \zonotope{e,t},\forall t\geq 0$.
\begin{lemma}[Error zonotope]\label{lemma:error_zonotope_t}
    \colored{Let $\zonotope{e,0}=\langle e_0,0\rangle$}. At time $t\ge 0$, the error zonotope is:
    \begin{equation}
    \begin{aligned}
        \zonotope{e,t} = (A_0+B_0K)^te_0 
        +\sum_{k=0}^{t-1}(A_0+B_0K)^k\zonotope{\tilde w,t-k-1},
    \end{aligned}
    \end{equation}
    with $\zonotope{\tilde w,t} \coloneqq \Delta A_0\bar x_t + \Delta B_0 \bar u_t   + \zonotope{w}$.
    Moreover, if \colored{$A_0+B_0K$ is Schur stable}, and if $(\bar x_t, \bar u_t)_{t\ge 0}$ is a \colored{uniformly bounded sequence}, then $\zonotope{e,t}$ is a \colored{uniformly bounded set} for any $t\geq 0$.
\end{lemma}
$\zonotope{e,t}$ represents the set of reachable errors at time $t$. The proof stems from the expression of $e_{t+1}$ and the fact that the pair  $(\bar x_t, \bar u_t)$, is bounded. The boundedness of this pair follows from the MPC formulation provided below. The idea is to solve a receding-horizon optimization problem that at each step bounds the nominal dynamics $\bar x_t$, so that $\bar x_t + \zonotope{e,k} \subseteq \zonotope{x}$, so as to guarantee that the true dynamics will belong to $\zonotope{x}$. Similarly, we also constrain the signal $\bar u_t$. 

There are two problems left to solve: (i) the zonotope $\zonotope{e,t}$ cannot be used in practice since the true matrices $(A_0,B_0)$ are unknown; (ii) we need to guarantee the stability of $A_0+B_0K$. Regarding the former problem, the idea is to derive a conservative approximation $\bar{\mathcal{Z}}_{e,t}$ of $\zonotope{e,t}$, so that $\zonotope{e,t} \subseteq  \bar{\mathcal{Z}}_{e,t}$. The latter problem can be solved by finding $K$ that is stabilizing for all $(A,B)$ in $\Sigma_\data$.  \vspace{4pt}

\noindent \textbf{Conservative approximation of the error zonotope.}
As already mentioned, since the pair $(A_0,B_0)$ is unknown, we cannot consider directly $\zonotope{e,t}$ in the optimization algorithm that we wish to solve. Therefore, we construct an approximation $\bar{\mathcal{Z}}_{e,t}$ of $\zonotope{e,t}$. First, observe that from \cref{lemma:error_zonotope_t}:
\begin{equation}
    \zonotope{e,t} =(A_0+B_0K)\zonotope{e,t-1} + \Delta A_0\bar x_t+ \Delta B_0 \bar u_t   + \zonotope{w}.
\end{equation}
Define $\matrixzonotope{\data, K} \coloneqq \matrixzonotope{\data}\begin{bmatrix}
    I_n \\ K
    \end{bmatrix}$ and \colored{$\matrixzonotope{\Delta} \coloneqq \matrixzonotope{\data}- \begin{bmatrix}\bar A & \bar B\end{bmatrix}$}. Then, we obtain the following approximation.
\begin{lemma}[Error zonotope approximation]
    Let $\bar{\mathcal{Z}}_{e,t}$ be defined as 
    \begin{equation}
    \bar{\mathcal{Z}}_{e,t} \coloneqq \matrixzonotope{\data, K}\bar{\mathcal{Z}}_{e,t-1} + \matrixzonotope{\Delta }\begin{bmatrix}
    \bar x_t \\ \bar u_t
    \end{bmatrix} + \zonotope{w},
    \end{equation}
    with $\bar{\mathcal{Z}}_{e,0}=\zonotope{e,0}$. Then, $\zonotope{e,t} \subseteq \bar{\mathcal{Z}}_{e,t}$ for $t\geq 0$.
\end{lemma}
\begin{proof}
We prove it by induction. Obviously it holds for $t=0$. For a fixed $t>0$ we observe that by construction of $\matrixzonotope{\Delta}$ it holds that $\Delta A_0 \bar x_t +\Delta B_0 \bar u_t \in \matrixzonotope{\Delta}(\langle \bar x_t, 0\rangle \times\langle \bar u_t, 0\rangle)$. Therefore
\colored{$\Delta A_0 \bar x_t +\Delta B_0 \bar u_t +\zonotope{w} \subseteq \matrixzonotope{\Delta }\left(\langle \bar x_t, 0\rangle \times\langle \bar u_t, 0\rangle\right) + \zonotope{w}$}. Using the induction step, since $\zonotope{e,t-1}\subseteq \bar{\mathcal{Z}}_{e,t-1}$, and $A_0+B_0K = \begin{bmatrix}
A_0 & B_0
\end{bmatrix}\begin{bmatrix}
I_n \\ K
\end{bmatrix}\in \matrixzonotope{\data}\begin{bmatrix}
I_n \\ K
\end{bmatrix}$, it follows that $(A_0+B_0K)\zonotope{e,t-1} \subseteq\matrixzonotope{\data, K}\bar{\mathcal{Z}}_{e,t-1}$.
\end{proof}
\colored{To guarantee the stability of the new error-zonotope, we need the following assumption that there exists a common quadratic Lyapunov function.
\begin{assumption}\label{assumption:stabilizability_m_data}
There exists  $P=P^\top, P\succ 0,$ such that  $\forall G\in\matrixzonotope{\data,K}$ the inequality $G^\top PG-P \prec 0$ is satisfied.
\end{assumption}
The previous assumption constraints the vertices of the convex set $\matrixzonotope{\data,K}$ to have a common Lyapunov function. The assumption can be relaxed by considering multiple Lyapunov functions, as in \cite[Thm. 8]{hu2010non} (which we omitted for brevity).}
Then we obtain the following stability result for $\bar{\mathcal{Z}}_{e,t}$.
\begin{lemma}[Stability of the error zonotope]\label{lemma:stability}
Given \cref{assumption:stabilizability_m_data} and $e_0=0$, then there exists a zonotope $\bar{\mathcal{Z}} \subset \mathbb{R}^n$ that satisfies: (i) $ \bar{\mathcal{Z}}_{e,t} \subset \bar{\mathcal{Z}} $ for every $t \geq 0$; (ii) $\bar{\mathcal{Z}} $ is an invariant set, i.e., for $e \in \bar{\mathcal{Z}} \Rightarrow \matrixzonotope{\data, K}e + \bar w\in \bar{\mathcal{Z}} $, for all $\tilde w \in \matrixzonotope{\Delta }(\zonotope{x} \times \zonotope{u}) + \zonotope{w}$.
\end{lemma}
\begin{proof}
\colored{Define the disturbance set at time $t$ as:
\begin{equation}
\begin{aligned}
    V &= \left\{ \Delta  \begin{bmatrix}
    \bar x \\ \bar u
    \end{bmatrix} +w: 
 \Delta \in \matrixzonotope{\Delta} , \bar x \in \zonotope{x}, \bar u \in \zonotope{u}, w \in \zonotope{w}\right\},
\end{aligned}
\end{equation}
and note that the set of reachable errors at time $t$ is $E_t = \left\{GE_{t-1}+V: G \in \matrixzonotope{\data, K} \right\},$ with $E_0 = \{0\}$. Using \cref{assumption:stabilizability_m_data}, and the result from \cite{goebel2006dual} (or see also \cite[Thm. 8]{hu2010non}), we find that for $V=0$ the linear difference inclusion $q_t\in \{G q_{t-1}: G\in \matrixzonotope{\data, K}\}$ is exponentially asymptotically stable for any $q_0\in \mathbb{R}^n$. Therefore, the result follows from the fact that the disturbance set $V$ is bounded and compact, with $0\in V$ (since $(\bar A,\bar B)\in \matrixzonotope{\data}$ and $0\in \zonotope{w}$).}
\end{proof}
\noindent \textbf{Optimization problem.} We are now ready to present our algorithm. Define $\theta=(\bar A, \bar B,K)$ to be the parameter of the problem. Then, the optimization problem is formulated as:
\begin{equation}\label{problem:robust_tz-ddpc}
    \begin{aligned}
    {\mathcal{P}}_N&(e_t, \bar x_t,\theta): \min_{\colored{\bar u_{0|t},\dots,\bar u_{N-1|t}}} \quad  \sum_{\colored{k=0}}^{\colored{N-1}} \ell_k(\bar x_{k|t},\bar u_{k|t})\\
    &\textrm{s.t.} \quad 
    \bar x_{k+1|t}=\bar A \bar x_{k|t}+ \bar B \bar u_{k|t},\quad \bar x_{\colored{0}|t}=\bar x_t,\\
    & \bar{\mathcal{Z}}_{e,k+1|t}= \matrixzonotope{\data,K}\bar{\mathcal{Z}}_{e,k|t} + \matrixzonotope{\Delta }\begin{bmatrix}
    \bar x_{k|t}\\ \bar u_{k|t}
    \end{bmatrix}+ \zonotope{w},\\
    & \bar{\mathcal{Z}}_{e,k|t}+\bar x_{k|t} \subseteq \zonotope{x},\quad \bar{\mathcal{Z}}_{e,\colored{0}|t}= e_t,\\
    & K\bar{\mathcal{Z}}_{e,k|t}+\bar u_{k|t} \subseteq \zonotope{u},\quad \colored{k=0,\dots,N-1}.
    \end{aligned}
\end{equation}
\begin{algorithm}[t] 
\caption{\textsc{TZ-DDPC: Tube-based Zonotopic Data-Driven Predictive Control}}
\begin{algorithmic}[1]\label{algo:tzddpc}
    \REQUIRE Data $\data$, zonotopes $(\zonotope{w}, \zonotope{x}, \zonotope{u})$, horizons $(N,M)$
    \smallskip
    \STATE Use $\data$ to compute $\matrixzonotope{\data}$, choose $(A_n,B_n) \in \matrixzonotope{\data}$ and feedback gain $K$. Set $\theta \gets (A_n,B_n,K)$.
    \STATE Set $t\gets 0, \bar x_t \gets x_t, e_t \gets 0 $.
    \REPEAT
        \STATE Solve $\mathcal{P}_N(e_t,\bar x_t,\theta)$ in problem (\ref{problem:robust_tz-ddpc}) to get $\bar u_t^\star = \{\bar u_{\colored{0}|t}^\star, \dots, \bar u_{\colored{N-1}|t}^\star\}$, $\bar x_t^\star = \{\bar x_{\colored{0}|t}^\star, \cdots, \bar x_{\colored{N-1}|t}^\star\}$. 
        \STATE Set $\bar x_{t+1}\gets\bar x_{1|t}^\star$. Apply control signal $u_t = K e_t +\bar u_{0|t}^\star$ and observe $e_{t+1} = x_{t+1}-\bar x_{t+1}$.
        \STATE Set $t \gets t+1$.
    \UNTIL{$t \leq M$}
\end{algorithmic}
\end{algorithm}
$ {\mathcal{P}}_N(e_t, \bar x_t,\theta)$ can be cast as a convex problem, and its solution yields the optimal control sequence $\bar u_t^\star = \{\bar u_{\colored{0}|t}^\star, \dots, \bar u_{\colored{N-1}|t}^\star\}$ and the associated optimal nominal state sequence $\bar x_t^\star = \{\bar x_{\colored{0}|t}^\star, \cdots, \bar x_{\colored{N-1}|t}^\star\}$.  By repeatedly solving this optimization problem, we obtain the receding-horizon procedure in Algorithm \ref{algo:tzddpc}.

Furthermore, it is straightforward to observe that \colored{if at time $t=0$ Alg. \ref{algo:tzddpc} is feasible  $\forall x_0 \in \zonotope{x}$, then it is feasible at every iteration $0\leq t\leq M$, and $\forall t\geq 0$ the system satisfies $x_t\in \zonotope{x}, u_t\in\zonotope{u}$  under the process noise $w_t\in \zonotope{w}$.}\vspace{4pt}

\colored{\noindent \textbf{Computational simplification.} We  propose a simple change to ease the computational burden of the algorithm by taking advantage of the stability induced by $K$.  Define the operator $T_{\data,K}\zonotope{} = \matrixzonotope{\data, K} \zonotope{}$, so that  $T_{\data,K}^n \zonotope{}= T_{\data,K} \left(T_{\data,K}^{n-1}\zonotope{}\right)$. Then, $\bar{\mathcal{Z}}_{e,t}$ can be recursively written as
\begin{equation}
    \bar{\mathcal{Z}}_{e,t} = T_{\data,K}^t \zonotope{e,0} +\sum_{k=0}^{t-1} T_{\data,K}^{k} \left[\matrixzonotope{\Delta }\begin{bmatrix}
    \bar x_{t-k-1}\\ \bar u_{t-k-1}
    \end{bmatrix} +\zonotope{w} \right].
\end{equation}
In light of \cref{lemma:stability}, the user may consider approximating the latter term in the equation  as
\begin{equation}\label{eq:simplification}
\mu \lambda^{k_0} \langle 0, I_n \rangle+ \sum_{k=0}^{k_0-1} T_{\data,K}^{k} \left (\matrixzonotope{\Delta }\begin{bmatrix}
    \bar x_{t-k-1}\\ \bar u_{t-k-1}
    \end{bmatrix} +\zonotope{w} \right),
\end{equation}
where $\lambda \in(0,1)$ depends on the spectrum of $P$, and $\mu>0$ depends on the bound of the disturbance term (see also \cite{bof2018lyapunov}). The parameter $k_0$ accounts only for the last $k_0$ disturbances.}\\

\subsubsection{Selection of a feedback gain $K$}\label{subsubsec:k} To conclude, we consider the problem of identifying a  gain $K$. For example, $K$ can be calculated by solving the following LMI
\begin{equation}\label{eq:ens}
    (A+BK)^\top P (A+BK)-P \prec 0, \quad \forall (A,B)\in \matrixzonotope{\data},
\end{equation}
for some symmetric $P\succ 0$. To solve the above problem, one only needs to consider the vertices of $\matrixzonotope{\data}$ \cite{ben2001lectures,vidyasagar2001probabilistic}, which, if approximated by an hypercube, has $2^{n(n+m)}$ vertices. \colored{However, this computation is feasible when $n$ and $m$ are not too large (for $n=5$, $m=1$, there are $2^{30}$ vertices).} Alternatively, we propose two methods based on random sampling.
We analyse the following  two  problems: (1) that of verifying that a given $K$ is stabilizing \colored{(\textit{see the technical report})}; (2) the problem of computing a stabilizing $K$.  \vspace{4pt}
\ifdefined\arxiv

\noindent \textbf{Verification of $K$ through random sampling.} Let $\mathbb{P}^N$ be the $N$-fold product of $\mathbb{P}$. We verify whether a given $K$ is stabilizing using a batch  $\omega_N = \{(A^{(1)}, B^{(1)}),\cdots, (A^{(N)}, B^{(N)})\} \in \matrixzonotope{\data}^N$ of  $N$ i.i.d. samples drawn according to $\mathbb{P}^N$ over $\matrixzonotope{\data}$. 
For a pair $(A,B)$, define $g_K(A,B) \coloneqq \mathbf{1}_{\rho(A+BK)\geq 1}$ to be a binary function that returns $1$ if $\rho(A+BK)\geq1$, where $\rho$ is the spectral radius. Similarly, define for the batch $\omega_N$, the function $ g_K(\omega_N) \coloneqq \max_{i=1,\dots,N} g_K(A^{(i)},B^{(i)})$. For a fixed  $\omega_N$, define the risk of violation  over $\Sigma_\data$ as:
\begin{equation}
R_K(\omega_N) = \mathbb{P}\left((A,B)\in \Sigma_{\cal D}, g_K(A,B)> g_K(\omega_N) \right).\end{equation}
\begin{lemma}[Robustness guarantee for a given $K$\ifdefined\arxiv\else$^{2}$\fi]\label{lemma:robustness_verification_k}
For a given accuracy $\varepsilon \in (0,1)$ and confidence $\delta \in (0,1)$, if $N \geq \ln(\frac{1}{\delta})/\ln(\frac{1}{1-\varepsilon})$, then with probability $1-\delta$ we have $R_K(\omega_N) \leq \varepsilon$, that is $\mathbb{P}^N(R_K(\omega_N) \leq \varepsilon) \ge 1- \delta$.
\end{lemma}\ifdefined\arxiv
\begin{proof}
Let $\tilde R_K(\omega_N)= \mathbb{P}(g_K(A,B) > g_K(\omega_N))$.
Using \cite[Thm. 3.1]{tempo1996probabilistic}, we can straightforwardly obtain
$
    \mathbb{P}(\tilde R_K(\omega_N) > \varepsilon) \leq (1-\varepsilon)^N
$. Since $\Sigma_\data \subseteq \matrixzonotope{\data}$ it follows that for a given $\omega_N$ the inequality $\tilde R_K(\omega_N) \geq R_K(\omega_N)$ holds, thus
$ \mathbb{P}^N(R_K(\omega_N) > \varepsilon) \leq  \mathbb{P}^N(\tilde R_K(\omega_N) > \varepsilon) \leq (1-\varepsilon)^N.$ The proof follows by considering the complement and setting $\delta \geq (1-\varepsilon)^N $.
\end{proof}\else\fi
Hence, if $K$ is stabilizing for a given batch $\omega_N$, \emph{i.e.}, $g_K(\omega_N)=0$, then with confidence $1-\delta$, the probability that $K$ does not stabilize $(A,B)\in \Sigma_\data$ is lower than $\varepsilon$.  
In comparison to classical Chernoff bounds, which scale as $1/\varepsilon^2$, the bound in \cref{lemma:robustnessK} scales as $1/\varepsilon$ (since $\ln(1/(1-\varepsilon))\approx \varepsilon $ for small $\varepsilon$), greatly reducing the number of required samples. 
Note that the probability measure $\mathbb{P}$ can be chosen by the user, and may encode the a-priori information she has about $(A_0,B_0)$. \vspace{4pt}
\else\fi

\noindent \textbf{Computation of $K$ through random sampling.}
\ifdefined\arxiv\else \colored{Let $\rho(X)$ be the spectral radius of $X$}, and $\mathbb{P}^N$ be the $N$-fold product of a probability measure $\mathbb{P}$ over $\matrixzonotope{\data}$. \fi  Given a batch \colored{$\omega_N = \{(A^{(i)}, B^{(i)})\}_{i=1}^N$} of  $N$ i.i.d. samples drawn according to $\mathbb{P}^N$ over $\matrixzonotope{\data}$,  the user can learn $K$ by solving the following LMI
\begin{equation*}
        \begin{bmatrix}
        X & AX+BZ\\ (AX+BZ)^\top & X
        \end{bmatrix} \succ 0, \quad \forall (A,B)\in \omega_N
    \end{equation*}
in $X\succ0, Z$. Then, one finds  $K= ZX^{-1}$. Since $K$  is function of $\omega_N$, then $K$ is a random variable. The following result provides probabilistic guarantees for $K(\omega_N)$.
\begin{lemma}[Robustness guarantee for $K(\omega_N)$\ifdefined\arxiv\else$^{2}$\fi]\label{lemma:robustnessK}
For a given accuracy $\varepsilon \in (0,1)$ and confidence $\delta \in (0,1)$, let $N \geq \frac{5}{\varepsilon}(\ln\frac{4}{\delta} + d \ln\frac{40}{\varepsilon})$ with $d=2nm \log_2 (2e n^2(n+1))$. Consider an i.i.d. sample $\omega_N$ from $\matrixzonotope{\data}$ sampled according to $\mathbb{P}^N$. Assume that $K=K(\omega_N)$ is computed according to $\omega_N$, and that $\rho(A+BK(\omega_N))<1$ for every $(A,B)\in \omega_N$. Then, with probability at-least $1-\delta$ we have 
\begin{equation}
    \mathbb{P}\left((A,B)\in \Sigma_\data, \rho(A+BK(\omega_N))\geq 1\right) \leq \varepsilon.
\end{equation}
\end{lemma}
\ifdefined\arxiv
\begin{proof} 
The sample complexity (the value of $N$ ensuring the desired probabilistic guarantees) can be found by computing the $\textrm{VC}$-dimension \cite{vidyasagar2001probabilistic} of $\mathcal{C}=\{\mathcal{S}(K), K\in \mathbb{R}^{m\times n}\}$, where $\mathcal{S}(K) =\{ (A,B) \in \matrixzonotope{\data}: \rho(A+BK) <1\}$. The stability test of $A+BK$ can be formulated using the Routh-Hurwitz criterion in the $s$-domain through a bilinear transform. The Routh-Hurwitz test consists of $n$ polynomial inequality, each with maximum degree $n(n+1)$ in the elements of $K$. Following the argument in \cite[Thm. 3]{vidyasagar2001probabilistic}, we apply \cite[Corollary 10.12]{vidyasagar2013learning}  with $l=nm, d=n(n+1)/2, s=n$, which yields $\textrm{VC}(\mathcal{C}) \leq 2nm \log_2 (2e n^2(n+1))$. The result follows by applying standard statistical learning arguments, for example by applying \cite[Corollary 4]{alamo2009randomized} and using a similar argument as in \cref{lemma:robustness_verification_k}.
\end{proof}\else\fi
\noindent \colored{First, note that the lemma holds for any learning procedure, not only the one that we propose. Secondly, observe that the sample complexity in \cref{lemma:robustnessK} scales as $\tilde O(nm)$ (we have hidden the logarithmic terms) for fixed $(\varepsilon,\delta)$ (\emph{cf.} checking the vertices of $\matrixzonotope{\data}$, which is exponential in $n$ and $m$).}

\section{Numerical simulations}
\ifdefined \arxiv

\begin{table}[t]
\setlength\arrayrulewidth{1pt}  
\centering
 \resizebox{0.5\textwidth}{!}{
\begin{tabular}{lccccccccccc}
\toprule
\multicolumn{1}{c}{} & \multicolumn{5}{c}{\textbf{Time [min]}} & \multicolumn{5}{c}{\textbf{Memory [MB]}} \\
\cmidrule(rl){2-6} \cmidrule(rl){7-11}
\multicolumn{1}{r}{$N$}  & {$1$} & {$2$} & {$3$} & {$4$} & {$5$}& {$1$} & {$2$} & {$3$} & {$4$}& {$5$}\\
\midrule
{\bf\textsc{ZPC}} & $0.02$ & $1.02$ & $*$ & $*$ & $*$ & $32$ & $1221$ & $34602$ & $*$& $*$  \\
{\bf \textsc{TZDDPC}} & $0.09$ & $0.13$ & $0.26$ & $1.99$ & $27.65$& $24$ & $67$ & $198$ & $2128$ & $30081$   \\
{\bf \textsc{TZDDPC}} - $k_0=2$ & $0.09$ & $0.13$ & $0.25$ & $1.87$ & $12.58$ & $22$ & $56$ & $195$& $1949$& $12122$\\
{\bf \textsc{TZDDPC}} - $k_0=1$ & $0.09$ & $0.13$ & $0.25$ & $1.25$ & $2.2$& $23$ & $57$ & $195$ & $1173$& $2153$   \\
\midrule
\multicolumn{1}{c}{} & \multicolumn{5}{c}{\textbf{Time [min]}} & \multicolumn{5}{c}{\textbf{Memory [MB]}} \\
\cmidrule(rl){2-6} \cmidrule(rl){7-11}
\multicolumn{1}{r}{$N$}  & {$6$} & {$7$} & {$8$} & {$9$} & {$10$}& {$6$} & {$7$} & {$8$} & {$9$}& {$10$}\\
\midrule
{\bf\textsc{ZPC}} & $*$ & $*$ & $*$ & $*$ & $*$ & $*$ & $*$ & $*$ & $*$& $*$  \\
{\bf \textsc{TZDDPC}} & $*$ & $*$ & $*$ & $*$ & $*$ & $*$ & $*$ & $*$ & $*$& $*$  \\
{\bf \textsc{TZDDPC}} - $k_0=2$ &  $*$ & $*$ & $*$ & $*$ & $*$ & $*$ & $*$ & $*$ & $*$& $*$ \\
{\bf \textsc{TZDDPC}} - $k_0=1$ & $3.21$ & $4.16$ & $5.14$ & $6.15$ & $7.16$ & $3113$ & $4073$ & $5076$ & $6034$& $6993$   \\
\bottomrule
\end{tabular}
}
\caption{{\footnotesize Max. amount of time/memory used by  the solver to build problem (\ref{problem:robust_tz-ddpc}) (over 5 runs). The symbol $*$ indicates that the solver needed more than $40$ minutes and/or over $35$ GB of memory to build the problem.}}
\label{tab:computational_efficiency}
\end{table}

\else

\begin{table}[t]
\setlength\arrayrulewidth{1pt}  
\centering
 \resizebox{0.5\textwidth}{!}{
 {\color{blue}
\begin{tabular}{lccccccccccc}
\toprule
\multicolumn{1}{c}{} & \multicolumn{5}{c}{\textbf{Time [min]}} & \multicolumn{5}{c}{\textbf{Memory [MB]}} \\
\cmidrule(rl){2-6} \cmidrule(rl){7-11}
\multicolumn{1}{r}{$N$}  & {$1$} & {$2$} & {$3$} & {$4$} & {$5$}& {$1$} & {$2$} & {$3$} & {$4$}& {$5$}\\
\midrule
{\bf\textsc{ZPC}} & $0.02$ & $1.02$ & $*$ & $*$ & $*$ & $32$ & $1221$ & $34602$ & $*$& $*$  \\
{\bf \textsc{TZDDPC}} & $0.09$ & $0.13$ & $0.26$ & $1.99$ & $27.65$& $24$ & $67$ & $198$ & $2128$ & $30081$   \\
{\bf \textsc{TZDDPC}} - $k_0=2$ & $0.09$ & $0.13$ & $0.25$ & $1.87$ & $12.58$ & $22$ & $56$ & $195$& $1949$& $12122$\\
{\bf \textsc{TZDDPC}} - $k_0=1$ & $0.09$ & $0.13$ & $0.25$ & $1.25$ & $2.2$& $23$ & $57$ & $195$ & $1173$& $2153$   \\
\bottomrule
\end{tabular}
}}
\caption{\colored{{\footnotesize Max. amount of time/memory used by  the solver to build problem (\ref{problem:robust_tz-ddpc}) (over 5 runs). The symbol $*$ indicates that the solver needed more than $40$ minutes and/or over $35$ GB of memory to build the problem.}}}
\label{tab:computational_efficiency}
\end{table}

\fi
We illustrate our method on a  double integrator affected by strong adversarial noise. In addition to that, we only make use of a small sample of data to construct $\matrixzonotope{\data}$. The choice of the system, as well as the presence of strong noise, and the small sample size,  make sure that   $\matrixzonotope{\data}$ contains unstable systems, so that we can verify the effectiveness of the method.
To handle the mathematical operations with zonotopes, we created a python library  \textsc{PyZonotope}\footnote{\textsc{PyZonotope}: \href{https://github.com/rssalessio/pyzonotope}{github.com/rssalessio/pyzonotope}}. The code for \textsc{TZ-DDPC}\footnote{\textsc{TZ-DDPC}: \href{https://github.com/rssalessio/TZDDPC}{github.com/rssalessio/TZDDPC}} was written in Python, and can be found on GitHub. \ifdefined\arxiv\else Please refer to the technical report\footref{techreport} for further implementation details. \vspace{4pt}\fi
\ifdefined\arxiv To implement the problem in (\ref{problem:robust_tz-ddpc}), the zonotope inclusion constraints are approximated by considering their right and left interval limits as in \cite{alanwar2022robust}. To reduce the complexity of the problem, the order of all the matrix zonotopes is reduced to $1$ using the box reduction method presented in \cite{kopetzki2017methods}.\vspace{4pt} \else\fi

\noindent \textbf{Numerical results.} The sampled double integrator is defined by the equations:
\begin{equation}
    x_{t+1} = \begin{bmatrix}
    1 & 1 \\ 0 & 1
    \end{bmatrix}x_t + \begin{bmatrix}0.5 \\ 1\end{bmatrix} u_t + w_t,\quad x_0=\begin{bmatrix}-5\\-2\end{bmatrix}.
\end{equation}
We chose a strong adversarial noise $w_t$, uniformly sampled from the vertices of $\zonotope{w}=\left\langle 0, \begin{bmatrix}0.1 & 0.05 \\ 0.05 & 0.1\end{bmatrix}\right\rangle$. The state zonotope is $\zonotope{x} = \langle \begin{bmatrix}-4, 0 \end{bmatrix}^\top, \diag(4, 2) \rangle$, while the control signal zonotope is $\zonotope{u} = \langle 0, 1 \rangle$. The cost function at any step $k$ is defined by $\ell(x,u) = \|x\|_2^2 + 10^{-2} |u|$. The matrix zonotope $\matrixzonotope{\data}$ was built using $T=100$ samples, collected using a standard normal distribution for $u_t$. \colored{The matrix $K=\begin{bmatrix}-0.561 & -1.385\end{bmatrix}$ was computed by approximately solving the optimization problem $K\in \{K: \rho(A+BK)<1, (A,B)\in \argmax_{(A,B)\in \matrixzonotope{\data}} \|A+BK\|_2\}$ through the use of concave programming \cite{shen2016disciplined}. }
\ifdefined \arxiv Finally, the solution was verified using Lemma \ref{lemma:robustness_verification_k} with $\varepsilon=10^{-2},\delta=10^{-5}$. \else \fi

In Fig. \ref{fig:double_integrator}, we compare results for \textsc{TZ-DDPC}  and \textsc{ZPC} with $M=12$. \textsc{ZPC} is computationally complex, and hence, we simulated \textsc{ZPC} with $N=2$. For fair comparison, we used the same value of  $N$ for \textsc{TZ-DDPC}. With the same data, and constraints, \textsc{ZPC} could not solve the problem without enlarging the size of $\zonotope{x}$ by approximately $25\%$. This constraint violation is also seen in Fig. \ref{fig:double_integrator}. \vspace{4pt}

\colored{\noindent \textbf{Computational complexity.} Additionally, in \cref{tab:computational_efficiency} we simulated the problem for different values of $N$, and evaluated the amount of time and memory needed by the solver to solve the convex problem. In addition to comparing \textsc{ZPC} with \textsc{TZDDPC}, we also evaluated \textsc{TZDDPC} with the simplification in \cref{eq:simplification} for $k_0=1,2$. Results show how \textsc{ZPC} cannot cope with larger horizons, whereas \textsc{TZDDPC} with $k_0=1$ can be used to efficiently solve the problem.}
We conclude that using a stabilizing matrix $K$ can help improve stability, and reduce the complexity of using zonotope-based methods.
\begin{figure}[t]
    \centering
    \includegraphics[width=\linewidth]{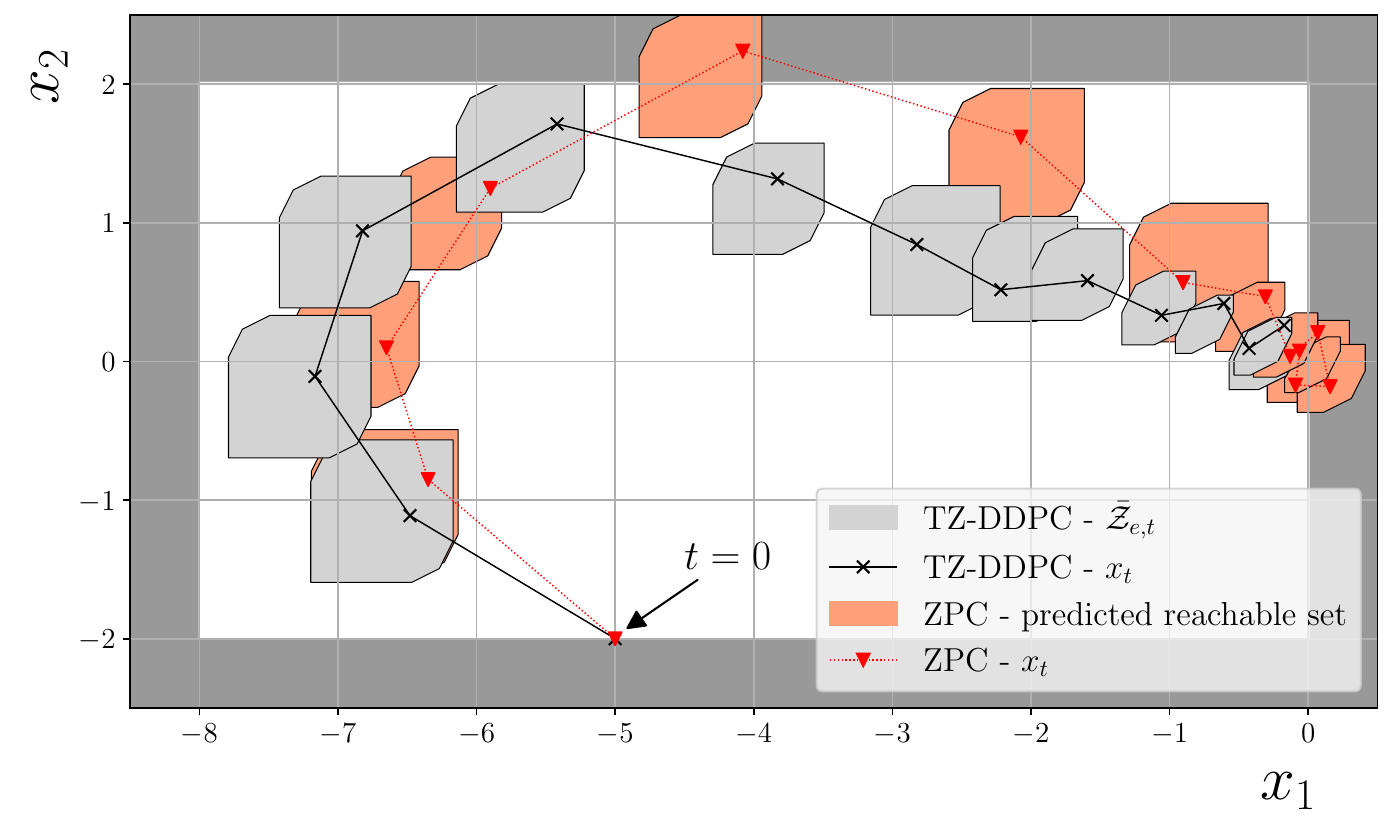}
    \caption{Double integrator: comparison of \textsc{TZ-DDPC} and \textsc{ZPC}. The gray area depicts the complement of $\zonotope{x}$.  As seen from the image, \textsc{ZPC} is not feasible for the original constraint zonotope $\zonotope{x}$, which was then enlarged by  $25\%$ to make the problem feasible for \textsc{ZPC}.}
    \label{fig:double_integrator}
\end{figure}

\section{Conclusion}
In this work we proposed a tube-based MPC formulation based on zonotopes to deal with generic convex loss functions, bounded process noise and uncertainties in the system matrices. Our method builds on \cite{alanwar2022robust}, and consists of two phases: (i) an offline data-collection phase that builds a set of possible system matrices that are consistent with the data; (ii) an online control phase that uses a tube-based MPC paradigm to robustly control the unknown linear system. We show how to guarantee stability of the resulting error zonotope, and provide probabilistic robustness guarantees for the stabilizing gain matrix $K$. \colored{Future venues of research include: find less conservative approximations of the error zonotope and/or use persistently exciting control signals to shrink the set of uncertain system matrices.}
\ifdefined\arxiv\else
\addtolength{\textheight}{-12cm}   
\fi
\bibliographystyle{plain} 
\bibliography{refs} 

\end{document}